\def\EXTENDED{}
\documentclass[runningheads]{llncs}

\usepackage{amsmath,amssymb,mathtools}
\usepackage{graphicx}
\usepackage[hidelinks]{hyperref}
\usepackage{booktabs}

\newcommand{\E}{\mathbb{E}}

\newcommand{\Prob}{\mathbb{P}}
\newcommand{\R}{\mathbb{R}}
\newcommand{\dt}{\Delta t}
\newcommand{\1}{\mathbf{1}}
\newcommand{\bound}{\mathrm{bound}}

\newif\ifextended
\ifdefined\EXTENDED\extendedtrue\else\extendedfalse\fi

\newcommand{\deferproof}[1]{%
  \par\smallskip\noindent\textit{Proof.}\ \ifextended
  See Appendix~#1.\else
  See Appendix~#1 of the extended version~\cite{bundi2026ext}.\fi\par\smallskip}

\title{Optimal Block Time for AMM\\ Liquidity Providers under Jump-Diffusion Prices%
\ifextended\thanks{Extended version of a contribution accepted at MARBLE 2026 and to be
published by Springer Nature in \emph{Mathematical Research for Blockchain Economy} (Lecture
Notes in Operations Research); it adds the proofs omitted from the published version. This is
not the Version of Record.}\fi}
\titlerunning{Optimal Block Time under Jump-Diffusion}

\author{Nils Bundi\orcidID{0000-0003-3576-3289}}
\authorrunning{N. Bundi}
\institute{Zurich University of Applied Sciences, School of Engineering,\\
Winterthur, Switzerland\\
\email{bund@zhaw.ch}}

\begin{document}
\maketitle

\begin{abstract}
Loss-versus-Rebalancing (LVR) is the dominant adverse-selection cost borne by liquidity providers on automated market makers. Under geometric Brownian motion, arbitrage profit scales with the probability of a profitable block, which vanishes as the block time $\dt \to 0$ --- the standing argument for ever-shorter blocks. Modeling the reference price instead as a jump-diffusion, I show that the constant-product LVR rate splits into a diffusion channel carrying the known multiplier $F(\gamma/(\sigma\sqrt{\dt}))$ and a jump channel $\lambda V \cdot G(\gamma;m,\delta^2)$ carrying no $\dt$, the two interacting only through an explicitly bounded remainder. The block schedule therefore governs only one channel. For \emph{symmetric} jump laws the jump channel is moreover an exact lower bound, $\ell(\dt) \ge \lambda V G > 0$, so the rate does not vanish as $\dt \to 0$, and descends \emph{slowly} as $\sqrt{\dt}$. At Ethereum's calibrated $12$-second slot the rate is $471$~bp/yr against a floor of $125$, so only three quarters of LP loss is schedule-addressable. At Solana's $400$~ms slot the jump channel already dominates. Netting the rate against per-block consensus cost, the LP-side optimal block time is invariant in pool size and in every jump parameter $(\lambda,m,\delta)$: jumps shift the level of LP loss but not the planner's marginal tradeoff. Volatility, the fee tier, and consensus cost set the optimum, near $8$~s. However, LVR is only one input to block-time welfare, so this bounds the LP-side contribution rather than settling the design question.

\keywords{Automated market makers \and Loss-versus-Rebalancing \and Jump-diffusion \and Block time \and Decentralized finance \and Blockchain protocol design}
\end{abstract}

\section{Introduction}\label{sec:intro}

Automated market makers (AMMs) intermediate the bulk of on-chain spot trading, yet the passive liquidity providers who fund them systematically cede value to better-informed arbitrageurs whenever the external price moves. How fast a blockchain produces blocks---the cadence at which arbitrageurs can act---directly shapes this loss, and the prevailing intuition holds that faster blocks are better. This paper asks whether price \emph{jumps} overturn that intuition, and finds that they impose a floor no block speed can breach.

Milionis et al.~\cite{milionis2022automated} formalize the adverse-selection cost of passive liquidity provision on a constant-function market maker (CFMM) as \emph{loss-versus-rebalancing} (LVR): the gap between a passive liquidity provider's (LP) mark-to-market return and the return of a continuously rebalanced portfolio holding the LP's instantaneous exposure. For a constant-product market maker (CPMM), a special case of the CFMM, with pool value $V(P)$ and external reference price $S_t$ following a GBM with volatility $\sigma$, the LVR rate equals $\sigma^2 V/8$. Hence, it is independent of the block-time, the rate at which arbitrageurs are able to re-align the pool. Milionis et al.~\cite{milionis2023automated} extend this to swap fees $\gamma$ and Poisson block arrivals of mean block-time $\dt$. In their low-fee fast-block regime, the LVR rate is multiplied by a strictly increasing function of $\sigma\sqrt{\dt}/\gamma$, so shorter blocks increase welfare.

This GBM-based prescription---``shorten blocks to reduce LVR''---is the implicit case for sub-second block times deployed on Solana, modern Ethereum L2s, and a host of app-chains. Empirically, however, crypto reference prices are not GBM but exhibit pronounced jumps \cite{scaillet2020high,chaim2019nonlinear,saef2024temporal}. A jump produces an instantaneous, $\dt$-independent price gap which is cleared in the very next block and the LP eats the concavity loss.

This paper asks: \emph{what does the LVR/block-time relationship look like once prices jump?} The answer is that it splits in two: one channel the block schedule governs, and one it cannot touch --- jumps drive the mispricing across the no-arbitrage band at a rate the protocol does not choose, so faster blocks clear each gap sooner without making it smaller. Calibrated to ETH/USDT, the split is stark: at Ethereum's $12$-second slot the rate is $471$~bp/yr against a jump floor of $125$; at Solana's $400$~ms the floor already dominates; and weighing LVR against Ethereum consensus cost puts the LP-side optimum near $8$~s, above today's sub-second targets of other chains.

\paragraph{Contributions.} For a CPMM under jump-diffusion with any finite-activity jump law --- Merton $(\mu,\sigma,\lambda,m,\delta)$ the running instance --- Theorem~\ref{thm:decomp} gives an additive zero-fee decomposition of the LVR rate into the classical diffusion term $\sigma^2 V/8$ and a nonnegative jump term $\tfrac{\lambda}{2}V\,\E[(e^{J/2}-1)^2]$ that carries no $\dt$. Reintroducing a swap fee $\gamma$ and discrete block time $\dt$, Theorem~\ref{thm:fee} shows that
\begin{equation}\label{eq:headline-fee}
\ell(\dt) \;=\; \tfrac{\sigma^2}{8}\,V \cdot F\!\left(\tfrac{\gamma}{\sigma\sqrt{\dt}}\right) \;+\; \lambda\,V \cdot G(\gamma; m, \delta^2) \;+\; \mathcal{E}(\dt),
\end{equation}
where $F:\R_+\to(0,1]$ is the trade-probability multiplier $P_{\mathrm{trade}}$ of Milionis et al.~\cite{milionis2023automated} (Lemma~\ref{lem:F}, Remark~\ref{rem:vsmmr}), $G = \tfrac{1}{8}\E[(|J|-\gamma)^2\1\{|J|>\gamma\}]$, and the remainder $\mathcal{E}$ --- the only place the two channels meet --- is bounded on both sides in Proposition~\ref{prop:error} under a stated mixing condition. For symmetric $\nu$ satisfying a mild aggregation hypothesis (Merton does), Theorem~\ref{thm:floor} establishes exactly, with no appeal to \eqref{eq:headline-fee}, that $\ell(\dt) \ge \lambda V G > 0$ at every $\dt$: an irreducible jump floor that no block-time reduction can eliminate, approached only as $\sqrt{\dt}$ (Corollary~\ref{cor:floor}). Finally, netting $\ell_0(\dt)$, the first two terms of \eqref{eq:headline-fee}, against pool-attributable block-production cost, Theorem~\ref{thm:planner} characterizes the LP-side planner's optimum: invariant in pool size, and in the jump parameters $(\lambda,m,\delta)$ to the order at which \eqref{eq:headline-fee} separates the channels --- jumps shift the welfare level but not the planner's marginal tradeoff.

\section{Related Work}\label{sec:related}

CFMMs are the dominant trading mechanism in decentralized finance, with the CPMM invariant first deployed at scale by Uniswap~\cite{adams2021uniswap}. Angeris et al.~\cite{angeris2020improved} characterize arbitrage equilibria, reserve dynamics and the convergence of pool prices to external references. Milionis et al.~\cite{milionis2022automated} introduce \emph{loss-versus-rebalancing} (LVR), the LP's underperformance against a continuously-rebalanced benchmark, and derive under GBM a measure-independent rate---independent of the LP's risk-neutral probability and of hedging considerations---with the full-range Uniswap-V2 specialization $\sigma^2 V/8$. Milionis et al.~\cite{milionis2023automated} extend this to proportional fees with Poisson-arrival blocks. Their key structural result, and the one this paper builds on, is that the log mispricing between pool and reference is ergodic with stationary law uniform on the no-arbitrage band and exponential tails outside it, so that arbitrage profits scale down from the frictionless LVR by the stationary probability of a profitable block---the first formal link between protocol-level block-time and LP loss. Nezlobin and Tassy~\cite{nezlobin2025lvr} extend this to general block-time distributions and Fritsch and Canidio~\cite{fritsch2024measuring} validate the scaling empirically across major Uniswap pools. Cartea, Drissi and Monga~\cite{cartea2023predictable} frame the counterpart notion of \emph{predictable loss} in continuous-time wealth dynamics and solve for optimal concentrated-liquidity provision against fee income and concentration risk, while Capponi and Jia~\cite{capponi2024price} micro-found the loss channel through a deposit game exhibiting a ``liquidity freeze'' regime.

On price dynamics, jump-diffusion has been standard in derivatives pricing since Merton~\cite{merton1976option}, with Kou's asymmetric double-exponential law~\cite{kou2002jump} and the Carr--Geman--Madan--Yor pure-jump L\'evy model~\cite{carr2002fine} as flexible alternatives, and Cont and Tankov~\cite{cont2003financial} the reference for the It\^o--L\'evy machinery used below. Empirically, crypto jumps are pervasive: Scaillet et al.~\cite{scaillet2020high} document frequent, clustered and order-flow-predictable jumps in Bitcoin at high frequency; Saef et al.~\cite{saef2024temporal} find in tick data from seven exchanges over April 2019--September 2021 that Bitcoin jumps on $\sim$$58\%$ and Ethereum on $\sim$$32\%$ of trading days, with systematic temporal clustering; and Chaim and Laurini~\cite{chaim2019nonlinear} document nonlinear dependence and tail co-movement across major cryptocurrencies. Cont~\cite{cont2001empirical} surveys the broader stylized facts of asset returns.

Block-time choice balances several welfare components beyond LP loss. On \emph{consensus cost}, Decker and Wattenhofer~\cite{decker2013information} identify block propagation delay as the primary cause of forks and Croman et al.~\cite{croman2016scaling} per-block bandwidth overhead as the binding scaling constraint. On \emph{finality}, Buterin~\cite{buterin2022paths} surveys paths to Ethereum single-slot finality, where BLS aggregation over committee attestations is the per-slot bottleneck. On \emph{decentralization and censorship resistance}, Benhaim et al.~\cite{benhaim2024scaling} show committee-based consensus attains equivalent security with committees one to two orders of magnitude smaller---legitimizing the architecture behind sub-second blocks, but exposing how small committees increase censorship vulnerability. On \emph{MEV redistribution}, Daian et al.~\cite{daian2020flash} document the maximal-extractable-value economy that block-time mediates: longer blocks expand per-block extractable value, shorter blocks fragment it and erode the searcher timing advantage. \emph{Oracle freshness}, the gap between on-chain feeds and live markets, scales directly with $\dt$. LP-side adverse selection, the subject of this paper, is one input among these.

Two gaps emerge, and this paper closes both. The LVR construct is well-developed under GBM but has not been extended to jump-diffusion reference prices despite the empirical prevalence of price discontinuities in crypto markets, and the LP-side literature lacks an explicit planner's optimum that exposes how the jump parameters enter, or fail to enter, that optimum.

\section{Model}\label{sec:model}

The model is that of Milionis et al.~\cite{milionis2022automated,milionis2023automated} with a single primitive changed: the reference price follows a Merton jump-diffusion rather than pure GBM, capturing the discontinuous price moves documented empirically in cryptocurrency markets.

\subsection{Price Process}
Let $(\Omega, \mathcal{F}, (\mathcal{F}_t), \Prob)$ be a filtered probability space supporting a standard Brownian motion $W$ and an independent Poisson random measure with intensity $\lambda \, dt \, \nu(dy)$, where $\nu$ is the distribution of the log jump size $J$, assumed to have finite activity and to satisfy $\int(e^{j/2}-1)^2\,\nu(dj)<\infty$; the quadratic results of Section~\ref{sec:fees} need only a finite second moment. Merton's $J \sim \mathcal{N}(m,\delta^2)$ is the running instance and the one calibrated in Section~\ref{sec:numerics}. Section~\ref{sec:zerofees} needs no more than the moment condition, while from Section~\ref{sec:fees} on I assume in addition that $\nu$ is \emph{symmetric} --- for Merton, $m=0$. Remark~\ref{rem:symmetry} isolates exactly where symmetry is used. The reference price $S_t$ follows
\begin{equation}\label{eq:merton}
\frac{dS_t}{S_{t-}} = \mu\,dt + \sigma \, dW_t + (e^{J_t} - 1)\,dN_t,
\end{equation}
with $N$ a Poisson process of intensity $\lambda$. Time is measured in years throughout, so $\sigma$, $\lambda$ and $\dt$ are all in annual units; block times are quoted in seconds in Section~\ref{sec:numerics} for readability only. The notation $S_{t-} := \lim_{s\uparrow t} S_s$ denotes the left limit, i.e., the pre-jump price at time $t$; using $S_{t-}$ rather than $S_t$ in the denominator makes the relative increment well-defined at jump times and ensures the integrand is predictable.

\subsection{CPMM}
A CPMM with reserves $(x_t, y_t)$ maintains the invariant $x_t y_t = L$. The marginal pool price is $P_t = y_t / x_t$, and the pool's holdings as functions of $P$ are $x(P) = \sqrt{L/P}$ and $y(P) = \sqrt{L P}$.
The pool value (in numeraire units) at marginal price $P$ is
$V(P) := y(P) + P\,x(P) = 2\sqrt{LP}$.
By the envelope identity, $V'(P) = x(P)$ and
$V''(P) = -\tfrac{1}{2}\sqrt{L}\,P^{-3/2} = -x(P)/(2P) < 0$.
Note that the pool value $V$ is concave and this concavity is the source of LVR. I work with the full-range Uniswap specialization throughout for concreteness, but per Milionis et al.~\cite[Remark~1]{milionis2022automated} the LVR construct depends only on a locally-smooth demand curve $x^*(P)$ and applies with no modification to general CFMMs and concentrated-liquidity AMMs (within active tick ranges).

\subsection{Arbitrageur}
A risk-neutral, capital-unconstrained arbitrageur observes the external price $S_t$ continuously and the pool price $P_t$. Trading on the pool incurs a proportional fee $\gamma \in [0,1)$ in log-price units (i.e., a swap effectively moves the marginal price by an extra factor $e^{\pm\gamma}$), so the no-arbitrage band is $\{|\log S - \log P| \le \gamma\}$. Blocks arrive as a Poisson process of rate $1/\dt$, so $\dt$ is the mean interblock time, and the arbitrageur can trade only when a block arrives. The reference price itself moves continuously between blocks: $S$ is an off-chain venue price, so the arbitrageur's information is continuous even though on-chain reserves settle only at block inclusion. Mempool visibility of pending transactions extends the same continuity to the on-chain order flow, so the pool price is not the only price a participant can observe between blocks. Being competitive, the arbitrageur trades myopically whenever it is profitable~\cite[\S3]{milionis2023automated}.

\subsection{LVR}
Following Milionis et al.~\cite{milionis2022automated}, define the LVR over $[0,T]$ as
\begin{equation}\label{eq:lvrdef}
\mathrm{LVR}(T) \;:=\; \int_0^T V'(P_{s-}) \, dP_s \;-\; \big[V(P_T) - V(P_0)\big].
\end{equation}
The first term is the return of a rebalancing strategy holding $V'(P_{s-}) = x(P_{s-})$ units of the risky asset at each instant; the second is the LP's mark-to-market return. The \emph{LVR rate} is $\ell(t) = \E_t[d\mathrm{LVR}_t]/dt$, evaluated at the current pool state.

In the continuous-arbitrage limit ($\gamma = 0$, $\dt \to 0$), the pool price tracks the reference price, $P_t \equiv S_t$, and \eqref{eq:lvrdef} is the It\^o--L\'evy remainder of $V(S_t)$.

\section{LVR under Jump-Diffusion: Zero-Fee Decomposition}\label{sec:zerofees}

\begin{theorem}[LVR decomposition]\label{thm:decomp}
Under the zero-fee continuous-arbitrage limit with $S_t$ following \eqref{eq:merton} and the CPMM pool value $V(P) = 2\sqrt{LP}$, the instantaneous LVR rate is
\begin{equation}\label{eq:thm1}
\ell \;=\; \underbrace{\frac{\sigma^2}{8} V(S_t)}_{\text{diffusion-LVR}} \;+\; \underbrace{\frac{\lambda}{2}\,V(S_t)\,\E[(e^{J/2} - 1)^2]}_{\text{jump-LVR}},
\end{equation}
and for $J \sim \mathcal{N}(m, \delta^2)$, $\E[(e^{J/2} - 1)^2] = e^{m + \delta^2/2} - 2 e^{m/2 + \delta^2/8} + 1 \geq 0$.
This expression is $0$ if and only if $J \equiv 0$.
\end{theorem}

\deferproof{A.1}

\begin{remark}[Measure]\label{rem:measure}
The drift enters \eqref{eq:lvrdef} only through $\int V'(S_{s-})\,dS_s$, which appears in both terms and cancels, so \eqref{eq:thm1} involves $\sigma$ and the jump measure alone and no normalization of $\mu$ is required. Under GBM the rate is in this sense measure-independent~\cite{milionis2022automated}; with jumps it is not, since an equivalent change of measure rescales $\lambda$ and reshapes $\nu$, so Theorem~\ref{thm:decomp} is a statement under the physical measure and calibration of the model is versus realized-price observations (Section~\ref{sec:numerics}).
\end{remark}

\begin{remark}\label{rem:generaljumps}
The proof depends on the jump distribution only through $\E[(e^{J/2}-1)^2]$. For any compound-Poisson jump measure $\nu$ with $\int (e^{j/2}-1)^2\,\nu(dj) < \infty$ and $\nu \neq \delta_0$, the jump-LVR rate is $\tfrac{\lambda V}{2}\int (e^{j/2}-1)^2\,\nu(dj) > 0$. Merton is illustrative; Kou's double-exponential law admits the same structure with different moments, as does any finite-activity $\nu$. The results of Section~\ref{sec:fees} likewise use the shape of $\nu$ only through symmetry and the aggregation hypothesis of Theorem~\ref{thm:floor}, not through Gaussianity.
\end{remark}

\begin{remark}\label{rem:dtinvariant}
Theorem~\ref{thm:decomp} does not involve $\dt$. At $\gamma=0$, the arbitrageur eliminates any pricing gap within the next block, so block-time refinement only redistributes \emph{when} the loss is realized, not its rate. The block-time becomes economically relevant only once fees create a no-arbitrage band.
\end{remark}

\section{Fees and Block Time}\label{sec:fees}

I now allow $\gamma > 0$ and finite $\dt > 0$. Write $z_t := \log(S_t/P_t)$ for the log mispricing between reference and pool. At a block arrival $\tau_i$ the arbitrageur trades iff $|z|>\gamma$ and moves the pool to the near edge of the no-arbitrage band, so $z_{\tau_i} = \bound(z_{\tau_i^-},-\gamma,\gamma)$; between blocks $P$ is frozen and $z$ inherits the dynamics of $\log S$. The mispricing is therefore a \emph{persistent state}: a trade leaves it at $\pm\gamma$, not at $0$, and absent a trade it carries forward. Writing $\mathcal{J}_i := \bound(z_{\tau_i^-},-\gamma,\gamma)-z_{\tau_i^-}$ for the arbitrage resets and $T_k$ for the jump times of $N$, the It\^o--L\'evy formula applied to $\log S$ under \eqref{eq:merton} gives
\begin{equation}\label{eq:ourz}
z_t \;=\; \big(\mu-\tfrac12\sigma^2\big)t + \sigma W_t + \underbrace{\sum_{k:\,T_k\le t} J_k}_{\text{exogenous, rate }\lambda} + \underbrace{\sum_{i:\,\tau_i\le t}\mathcal{J}_i}_{\text{endogenous, rate }1/\dt}.
\end{equation}
The mispricing is driven by two point processes: the exogenous price jumps, which push $z$ \emph{away} from the band at a rate the protocol cannot choose, and the block-triggered resets, which push it back. Every result below follows from that asymmetry\ifextended\ (Figure~\ref{fig:zpath}, Appendix~\ref{app:figs}, illustrates)\else\ (illustrated in Appendix~B of the extended version~\cite{bundi2026ext})\fi. As in~\cite[Assumption~2]{milionis2023automated} I set $\mu = \tfrac{1}{2}\sigma^2 - \lambda m$.

At a block with pre-trade mispricing $z$ the pool's marginal log-price moves by $\xi = \mathrm{sign}(z)(|z|-\gamma)_+$, and by the computation in the proof of Theorem~\ref{thm:decomp} the LP's loss is $h(\xi) := \tfrac{V}{2}(e^{\xi/2}-1)^2$. Blocks being Poisson and independent of $z$, the pre-block law of $z$ is its stationary law $\pi$, and
\begin{equation}\label{eq:rate-stationary}
\ell(\dt) \;=\; \frac{1}{\dt}\,\E_{\pi}\big[h(\xi)\big], \qquad h(\xi) \;=\; \tfrac{V}{8}\xi^2 \;+\; \tfrac{V}{2}\,\rho(\xi), \quad \rho(\xi) := (e^{\xi/2}-1)^2 - \tfrac{\xi^2}{4}.
\end{equation}
Two facts about the split in \eqref{eq:rate-stationary} carry the results below. The quadratic surrogate $\tfrac{V}{8}\xi^2$ is a \emph{convex} function of $z$, which $h$ itself is not --- a downward move of arbitrary size costs at most $V/2$ --- and convexity is what Theorem~\ref{thm:floor} needs. And the curvature residual $\rho$ is signed under symmetry: pairing $\pm\xi$,
\begin{equation}\label{eq:rhosign}
\tfrac12\big[\rho(\xi)+\rho(-\xi)\big] \;=\; \cosh\xi - 2\cosh\tfrac{\xi}{2} + 1 - \tfrac{\xi^2}{4} \;=\; \sum_{n\ge2}\frac{1-2^{1-2n}}{(2n)!}\,\xi^{2n} \;\ge\; 0,
\end{equation}
so for any symmetric law of $\xi$ the exact loss is at least the quadratic surrogate, with leading term $\tfrac{7}{192}\xi^4$. Both statements are exact, for every $\xi\in\R$.

\begin{remark}[Driftlessness and symmetry]\label{rem:symmetry}
Setting $\mu = \tfrac{1}{2}\sigma^2 - \lambda m$ makes $\log S_t$ driftless. Additionally setting $m=0$ makes $z$ symmetric. The two properties are used differently below: Driftlessness is a modelling restriction rather than a normalization, since the rate at $\gamma>0$ depends on the full law of $z$ (its cost is small, as a non-zero drift $\Delta$ tilts the stationary law by only $2\Delta/\sigma^2$~\cite[Thm.~8]{milionis2023automated}, worth under $0.01\%$ of $\ell$ for any $|\Delta|\le1.6$/yr, which covers the empirical ETH drift by a wide margin); Symmetry of $\nu$ is used only to sign the remainder and to prove the floor, and it is symmetry that is needed, not Gaussianity: Theorems~\ref{thm:decomp}, \ref{thm:fee} and \ref{thm:planner} and Lemmas~\ref{lem:F}--\ref{lem:G} hold for arbitrary $\nu$, while Proposition~\ref{prop:error} and Theorem~\ref{thm:floor} require it. Everything the block schedule governs --- the separated rate $\ell_0$, the planner's cubic and its invariances --- is therefore general; symmetry enters only in transferring those properties to the exact rate $\ell$.
\end{remark}

\subsection{Diffusion Contribution with Fees}

Everything therefore turns on $\pi$. For $\lambda=0$ it is given in closed form by~\cite[Thm.~1]{milionis2023automated}: with $\theta := \sqrt{2/\dt}/\sigma$ and $\eta := \theta\gamma = \sqrt{2}\gamma/(\sigma\sqrt{\dt})$,
\begin{equation}\label{eq:pi0}
p(x) \;=\; c \ \ (|x|\le\gamma), \qquad p(x) \;=\; c\,e^{-\theta(|x|-\gamma)} \ \ (|x|>\gamma), \qquad c=\frac{\theta}{2(1+\eta)},
\end{equation}
uniform on the band with mass $\eta/(1+\eta)$ and exponential tails of mass $\tfrac12(1+\eta)^{-1}$ each, so $\Prob_\pi(|z|>\gamma) = (1+\eta)^{-1}$. Two features of \eqref{eq:pi0} carry the argument: the band holds most of the mass when blocks are fast, which is why the multiplier below vanishes; and the tails decay on the scale $1/\theta = \sigma\sqrt{\dt/2}$, so a jump with $|J|\gg\gamma$ lands far outside both band and tail and is cleared in full at the next block however small $\dt$ is. With $\lambda>0$ the stationary law solves an integro-differential rather than an ordinary differential equation, the compound-Poisson generator adding a non-local term of relative order $\lambda\dt$; that perturbation is carried through Proposition~\ref{prop:error}\ifextended\ and illustrated in Figure~\ref{fig:statlaw}\else\ and illustrated in Appendix~B of the extended version~\cite{bundi2026ext}\fi.

\begin{lemma}[Diffusion LVR with fees]\label{lem:F}
Let $\lambda=0$. For the quadratic surrogate in \eqref{eq:rate-stationary} the LVR rate is exactly
\begin{equation}\label{eq:Fresult}
\ell_{\mathrm{diff}}(\dt) \;=\; \tfrac{\sigma^2}{8}\,V \cdot F(\kappa),
\qquad
F(\kappa) \;=\; \frac{1}{1+\sqrt{2}\,\kappa} \;=\; \Prob_\pi\big(|z|>\gamma\big),
\qquad \kappa \;:=\; \frac{\gamma}{\sigma\sqrt{\dt}},
\end{equation}
the stationary probability that an arriving block offers a profitable trade. $F$ maps $\R_+$ onto $(0,1]$, satisfies $F(0)=1$ and $F(\kappa)\sim1/(\sqrt2\kappa)$, and is strictly decreasing, with $F'(\kappa) = -\sqrt2\,(1+\sqrt2\,\kappa)^{-2} < 0$.
\end{lemma}

\deferproof{A.2}

Fritsch and Canidio~\cite{fritsch2024measuring} test this scaling empirically.

\begin{remark}[Relation to Milionis et al.]\label{rem:vsmmr}
Under a GBM reference price the mispricing is a regulated Brownian motion driven by the single point process of block arrivals, whose marks always push $z$ \emph{towards} the band~\cite[Eq.~(9)]{milionis2023automated}. Equation~\eqref{eq:ourz} adds a second, exogenous driver whose marks push $z$ away from it; the floor of Corollary~\ref{cor:floor} is the consequence. The per-trade loss $\tfrac{V}{8}\xi^2$ entering \eqref{eq:rate-stationary} is their Lemma~2 profit up to $e^{\gamma} = 1+O(\gamma)$, and \eqref{eq:Fresult} is exact for the quadratic loss.
\end{remark}

\subsection{Jump Contribution with Fees}

When the reference price jumps by $J$, the arbitrageur captures the concavity loss only if the jump exceeds the fee threshold (otherwise the no-arbitrage band absorbs it without realization).

\begin{lemma}[Jump LVR with fees]\label{lem:G}
Under \eqref{eq:merton}, and for the same quadratic surrogate underlying \eqref{eq:Fresult}, the separated contribution of jumps to the LVR rate is $\lambda V \cdot G_\nu(\gamma)$, where
\begin{equation}\label{eq:G}
G_\nu(\gamma) \;:=\; \tfrac{1}{8}\,\E_\nu\!\Big[(|J| - \gamma)^2 \cdot \1\{|J| > \gamma\}\Big],
\end{equation}
a functional of $\nu$ alone; I write $G(\gamma;m,\delta^2)$ for its Merton value.
For $J \sim \mathcal{N}(m,\delta^2)$, $G$ is strictly positive whenever $\Prob(|J|>\gamma) > 0$, satisfies $G(0;m,\delta^2) = (m^2 + \delta^2)/8$, and is strictly decreasing in $\gamma$.
\end{lemma}

\deferproof{A.3}

A closed-form evaluation for symmetric Merton jumps ($m=0$) follows from the truncated-normal second-moment identity
$\Psi(x) := \E\big[(|Z|-x)_+^2\big] = 2\big[(1+x^2)\Phi(-x) - x\varphi(x)\big]$ for $Z\sim\mathcal N(0,1)$.
Letting $\zeta := \gamma/\delta$, this yields
\begin{equation}\label{eq:Gclosed}
G(\gamma; 0, \delta^2) \;=\; \tfrac{\delta^2}{4}\,\big[(1+\zeta^2)\,\Phi(-\zeta) - \zeta\,\varphi(\zeta)\big] \;=\; \tfrac{\delta^2}{8}\,\Psi(\zeta).
\end{equation}
The jump term thus carries a fee discount of its own, $\Psi$ at the jump-size scale $\delta$; crucially it is $\Psi$, not $F$, and it does not involve $\dt$, because the size of a jump is set by $\delta$ rather than by how long the block lasted. At $\gamma=0$, $G(0;0,\delta^2) = \delta^2/8$. For asymmetric Merton ($m \neq 0$) and other jump distributions, $G$ admits analogous truncated-moment representations.

\begin{remark}[The curvature residual in closed form]\label{rem:curvature}
Lemmas~\ref{lem:F} and~\ref{lem:G} evaluate the quadratic surrogate of \eqref{eq:rate-stationary}. The residual $\tfrac{V}{2}\E_\pi[\rho(\xi)]$ they omit is available in closed form on both channels. On the diffusion channel the stationary tails are exponential with rate $\theta$, so the exact-$h$ rate is $\tfrac{V}{2\dt}\Lambda(\theta)/(1+\eta)$ with $\Lambda(\theta) = \theta^2/(\theta^2-1) - 2\theta^2/(\theta^2-\tfrac14) + 1$, and
$2\theta^2\Lambda(\theta) = 1 + \tfrac{7}{4}\theta^{-2} + O(\theta^{-4}) = 1 + \tfrac{7}{8}\sigma^2\dt + O((\sigma^2\dt)^2)$,
a relative correction of $2.2\times10^{-7}$ at $\dt=12$~s that vanishes as $\dt\to0$. On the jump channel the same substitution turns $G$ into a truncated Gaussian moment-generating expression $G^{h}$ in $\Phi$, with $G^{h}/G - 1 = \tfrac{7}{16}\delta^2 + O(\delta^4) = 1.6\times10^{-4}$ at $\delta=0.0192$. Both corrections are nonnegative by \eqref{eq:rhosign} whenever $\pi$ is symmetric, and both are carried as exactly-evaluated terms of $\mathcal{E}$ in Proposition~\ref{prop:error} --- so nothing below requires $\delta$ or $\sigma\sqrt{\dt}$ to be small. The unbounded support of $J$ is likewise immaterial: $\E[\rho(\xi)]$ is finite.
\end{remark}

\begin{remark}\label{rem:lvrequivalence}
Equation \eqref{eq:Gclosed} reveals a clean equivalence: to leading order in $\delta$, a Poisson stream of jumps with standard deviation $\delta$ contributes the same LVR rate as a diffusion of effective volatility $\sigma_J := \delta\sqrt{\lambda}$, both giving $\sigma_J^2 V/8$ at $\gamma=0$. Fees break it asymmetrically: the diffusion suffers the $F(\gamma/(\sigma\sqrt{\dt}))$ discount, which vanishes as $\dt \to 0$, while jumps suffer only the $\dt$-invariant $\Psi(\gamma/\delta)$ discount. The jump floor's persistence is the formal expression of this asymmetry.
\end{remark}

\subsection{Main Formula}

Combining the diffusion contribution~\eqref{eq:Fresult} with the jump contribution of Lemma~\ref{lem:G} yields the LVR rate as a sum of a $\dt$-dependent diffusion term and a $\dt$-invariant jump floor.

\begin{theorem}[LVR rate under jump-diffusion with fees]\label{thm:fee}
Under \eqref{eq:merton} with swap fee $\gamma \in [0,1)$ and mean block time $\dt > 0$, the LVR rate of a CPMM satisfies
\begin{equation}\label{eq:main}
\boxed{\;\ell(\dt) \;=\; \ell_0(\dt) + \mathcal{E}(\dt), \qquad \ell_0(\dt) \;:=\; \tfrac{\sigma^2}{8}\,V \cdot F\!\Big(\tfrac{\gamma}{\sigma\sqrt{\dt}}\Big) \;+\; \lambda \, V \cdot G(\gamma; m, \delta^2),\;}
\end{equation}
with $F$ as in \eqref{eq:Fresult} and $G$ as in \eqref{eq:G}. Equation \eqref{eq:main} is a definition of $\mathcal{E}$ and carries no content until $\mathcal{E}$ is controlled; Proposition~\ref{prop:error} bounds it on both sides. I call $\ell_0$ the \emph{separated rate}: it is the object in which the two channels do not interact, and the properties established for it below are properties of $\ell_0$ unless stated otherwise. The one property that transfers to $\ell$ itself with no remainder is the floor, Theorem~\ref{thm:floor}.
\end{theorem}

\begin{proposition}[Two-sided bound on $\mathcal{E}$]\label{prop:error}
Let $\nu$ be symmetric, and let $n_*$ denote the least $n$ with $\sup_{m,m'}\|K_0^{n}(m,\cdot)-K_0^{n}(m',\cdot)\|_{\mathrm{TV}}\le\tfrac12$, where $K_0$ is the post-block kernel at $\lambda=0$. Assume the mixing condition
\begin{equation}\label{eq:mix}
n_* \;\le\; C\,(1+\eta^2), \qquad C=2.
\end{equation}
Only $\eta$ enters \eqref{eq:mix}; it is checked numerically over every block time, volatility and Uniswap fee tier considered here. It is used solely to bound the finite-$\lambda$ remainder in the separated expression. With $g(d):=(|d|-\gamma)_+^2$, the separated rate $\ell_0$ differs from the exact stationary rate \eqref{eq:rate-stationary} in four ways: it charges $\E[g(J)]$ in place of $\E[g(M+X+J)]$ on the jump term, it charges that term at the full jump rate $\lambda$ rather than the thinned single-jump block rate $\lambda(1+\lambda\dt)^{-2}$, it evaluates the diffusion term against the $\lambda=0$ law \eqref{eq:pi0} and block clock, and it replaces $h$ by its quadratic surrogate. Each is signed, so the remainder is bounded asymmetrically (with $\delta^2 := \E_\nu[J^2]$ for non-Merton $\nu$):
\begin{multline*}
-\underbrace{\tfrac52\,\lambda\dt\,\ell_{\mathrm{diff}}(\dt)}_{\text{block clock}} - \underbrace{\tfrac{\lambda V}{4}\big(2\gamma^2+\sigma^2\dt\big)}_{\text{stationary law}} - \underbrace{2\lambda^2\dt\,V G_\nu(\gamma)}_{\text{thinning}}
\;\le\; \mathcal{E}(\dt)\\
\;\le\;
\underbrace{\tfrac{\lambda V}{8}\big(2\sigma^2\dt + \gamma^2\big)}_{\text{interaction}}
+ \underbrace{\tfrac{\lambda V}{4}\big(2\gamma^2+\sigma^2\dt\big)}_{\mathclap{\text{stationary law}}} + \underbrace{\mathcal{E}_{\mathrm{curv}}}_{\mathclap{\text{curvature}}} + \underbrace{\tfrac{V\lambda^2\dt}{8}\big(\gamma^2+3\sigma^2\dt+2\delta^2\big)}_{\mathclap{\text{multi-jump}}},
\end{multline*}
where the interaction term uses only $|M|\le\gamma$, which holds pathwise since a block leaves the pool inside the band, and $\mathcal{E}_{\mathrm{curv}} = \tfrac{V}{2}\E_\pi[\rho(\xi)] \ge 0$ is the curvature residual, evaluated exactly in closed form in Remark~\ref{rem:curvature}. The stationary-law term covers the one place $\ell_{\mathrm{diff}}$ is evaluated against \eqref{eq:pi0} rather than the true law of the carried state, which jump feedback shifts. Its leading piece $\lambda V\gamma^2/2$ does not vanish with $\dt$, because the band relaxes on the timescale $\gamma^2/\sigma^2$ whatever the block schedule, so an $O(\lambda\dt)$ perturbation per block accumulates over $O(\dt^{-1})$ blocks. The thinning term is the deficit of charging every jump at rate $\lambda$ when a fraction $1-(1+\lambda\dt)^{-2}$ of them land in intervals with other jumps, whose aggregate losses appear, nonnegative, in the last term of the upper side. At the calibration of Section~\ref{sec:numerics} this gives $\mathcal{E}(\dt) \in [-0.653,\,+0.849]$~bp/yr at $\dt=12$~s and $[-0.355,\,+0.464]$ at $50$~ms, never exceeding $0.29\%$ of the rate over the deployed range $50$~ms--$12$~s. Evaluating \eqref{eq:rate-stationary} by quadrature puts the realised remainder at $+0.165$ and $+0.052$~bp/yr respectively, inside the bound and of the predicted sign. As $\dt\to0$ the block-clock term vanishes while the interaction tends to $\lambda V\gamma^2/8$, so $\ell(\dt) - \lambda VG_\nu(\gamma)$ is squeezed into $[0,\,0.463]$~bp/yr, at most $0.37\%$ of the floor.
\end{proposition}

\deferproof{A.4}

The floor claim holds exactly, it is not subject to $\mathcal{E}$ and Proposition~\ref{prop:error}, and the argument never touches the stationary law --- which is what makes it available, since with $\lambda>0$ that law has no closed form. Two limits should be kept apart here: $\dt\to0$ shortens the interval between blocks, but the jump law is exogenous to the block schedule, so it does not shrink $J$. Fast blocks do not imply small jumps, and the floor exists precisely because of that.

\begin{theorem}[Exact jump floor]\label{thm:floor}
Let $\nu$ be symmetric and satisfy the aggregation hypothesis $\E[g(\sum_1^k J)] \ge k\,\E[g(J)]$ for every $k\ge1$, as Merton with $m=0$ does. Then for every $\dt>0$,
\[
\ell(\dt) \;\ge\; \lambda\,V \cdot G_\nu(\gamma) \;=\; \inf_{\dt'>0}\,\ell_0(\dt') \;>\; 0,
\]
with no approximation and no remainder, the middle equality being immediate from $F>0$ and $F(\kappa)\to0$. The floor that the separated rate \eqref{eq:main} locates is therefore a floor for the exact rate, not an artifact of the separation.
\end{theorem}

\deferproof{A.5}

\begin{corollary}[Approach to the floor]\label{cor:floor}
The separated rate $\dt \mapsto \ell_0(\dt)$ is strictly increasing on $(0,\infty)$, with
\[
\lim_{\dt\to 0} \ell_0(\dt) = \lambda V \cdot G(\gamma; m, \delta^2), \qquad
\lim_{\dt\to\infty} \ell_0(\dt) = \tfrac{\sigma^2}{8}V + \lambda V \cdot G(\gamma; m, \delta^2),
\]
and the floor is approached at the algebraic rate $\ell_0(\dt) - \lambda VG \sim \sigma^3V\sqrt{\dt}/(8\sqrt2\,\gamma)$. For the exact rate, $\lambda VG \le \liminf_{\dt\to0}\ell(\dt) \le \limsup_{\dt\to0}\ell(\dt) \le \lambda VG + \lambda V\gamma^2/8 + \lambda V\gamma^2/2 + \mathcal{E}_{\mathrm{curv}}$, the lower bound by Theorem~\ref{thm:floor} and the upper by Proposition~\ref{prop:error}; the two differ by $0.37\%$ of the floor at the calibration of Section~\ref{sec:numerics}. Monotonicity likewise transfers to $\ell$, under \eqref{eq:mix}, wherever the separated increment clears the width of the remainder envelope: $\ell(\dt_2)>\ell(\dt_1)$ for any $\dt_1<\dt_2$ with $\ell_0(\dt_2)-\ell_0(\dt_1) > 1.50$~bp/yr, which holds for every pair reported in Section~\ref{sec:numerics}, and quadrature of \eqref{eq:rate-stationary} confirms monotonicity of $\ell$ itself across that range.
\end{corollary}

\deferproof{A.6}

\begin{remark}[The floor as un-recoverable LVR]\label{rem:feerecovery}
In the accounting of~\cite[Thms.~3--4]{milionis2023automated}, arbitrage profit and fee income sum to the frictionless LVR of Theorem~\ref{thm:decomp}: fees \emph{split} LVR rather than reduce it, and $\ell(\dt)$ is the arbitrage half, LP loss net of the fees LPs receive back. On the diffusion channel that split is $F(\kappa):1-F(\kappa)$ with $F(\kappa)\to0$, so shorter blocks return diffusion-LVR to LPs as fee income. On the jump channel \eqref{eq:Gclosed} makes it $\Psi(\gamma/\delta):1-\Psi(\gamma/\delta)$, a function of $\gamma/\delta$ alone with no $\dt$ in it, because a jump exceeding the band is cleared at the next block whatever the block rate. Fees recover only $1-\Psi(\gamma/\delta)$ of jump-LVR---$4.1\%$ at the calibration below---and the rest is the floor. The lever against it is $\gamma/\delta$, i.e.\ fee-tier design, not block time.
\end{remark}

\section{The Planner's Block-Time Optimum}\label{sec:opt}

Shorter blocks impose real costs that must be funded somehow. Each block requires fixed per-block work---block gossip and propagation~\cite{decker2013information}, BLS signature aggregation over committee attestations~\cite{buterin2022paths}, and view-change/commit overhead in committee-based consensus~\cite{benhaim2024scaling}---so the true \emph{validator operating cost} is a resource cost of hardware, bandwidth, and energy that scales with the number of blocks produced. Since this resource cost is not directly observable, I rely on the chain's native-token issuance as a proxy for it. Because issuance is a chosen compensation policy rather than a measured cost, it overstates the resource cost where it acts as a growth subsidy and understates it where fees and MEV fund validators~\cite{buterin2022paths,beccuti2025staking}. This bias only shifts the level of the cost, an effect discussed in Section~\ref{sec:disc}. I thus adopt the simplest specification with fixed per-block cost rate $R(\dt) = r/\dt$, calibrating the per-block chain-level cost $r$ from observed issuance.

A representative pool with TVL $V$ inherits a share of the chain's marginal resource cost in proportion to value secured: $c(\dt) = R(\dt) \cdot V/V_{\text{chain}} = c/\dt$ where $c := r \cdot V/V_{\text{chain}}$ and $V_{\text{chain}}$ is the chain's native-token market capitalization. The LP-side planner chooses $\dt$ to minimize the total per-unit-time cost rate combining the LVR rate from Theorem~\ref{thm:fee} with the per-pool consensus cost,
$W(\dt) = \ell_0(\dt) + c/\dt$,
written against the separated rate $\ell_0$ of \eqref{eq:main}. Theorem~\ref{thm:planner} is therefore a statement about $\ell_0$; Remark~\ref{rem:exactfoc} quantifies what the remainder does to it.

\begin{theorem}[Optimal block time]\label{thm:planner}
Let $c > 0$. Any interior critical point $\dt^{\mathrm{opt}}$ of $W(\dt)$ satisfies
$\ell_0'(\dt^{\mathrm{opt}}) = c/(\dt^{\mathrm{opt}})^2$.
With $\kappa = \gamma/(\sigma\sqrt{\dt})$, $d\kappa/d\dt = -\kappa/(2\dt)$, and the derivative in Lemma~\ref{lem:F},
\begin{equation}\label{eq:focexplicit}
\ell_0'(\dt) \;=\; \tfrac{\sigma^2 V}{8}\,F'(\kappa)\cdot\Big(\!-\tfrac{\kappa}{2\dt}\Big)
\;=\; \frac{\sigma^2 V}{8}\cdot\frac{\sqrt2\,\kappa}{2\dt\,(1+\sqrt2 \kappa)^2} \;>\; 0 .
\end{equation}
The first-order condition, after substituting $\dt = \gamma^2 / (\sigma^2 \kappa^2)$ and writing $\eta := \sqrt2\,\kappa$, reduces to a cubic:
\begin{equation}\label{eq:planner}
\boxed{\;\eta^{\mathrm{opt}}\big(1+\eta^{\mathrm{opt}}\big)^{2} \;=\; \frac{V\gamma^2}{8c},
\qquad
\dt^{\mathrm{opt}} \;=\; \frac{2\gamma^2}{\sigma^2\,(\eta^{\mathrm{opt}})^2}.\;}
\end{equation}
Since $\eta\mapsto\eta(1+\eta)^2$ is a strictly increasing bijection of $(0,\infty)$ onto itself, \eqref{eq:planner} admits a unique positive solution for every $c>0$, and $W$ attains its minimum there. The root is available in closed form by Cardano's formula: with $A := V\gamma^2/(8c)$, $B := A/2 + 1/27$ and $u := B + \sqrt{B^2 - 1/729}$,
$\eta^{\mathrm{opt}} = u^{1/3} + \tfrac{1}{9}u^{-1/3} - \tfrac{2}{3}$.
\end{theorem}

\deferproof{A.7}

\begin{remark}\label{rem:cs}
Differentiating \eqref{eq:planner} yields four sensitivities. First, $\partial \dt^{\mathrm{opt}}/\partial \lambda = \partial \dt^{\mathrm{opt}}/\partial m = \partial \dt^{\mathrm{opt}}/\partial \delta = 0$: the jump parameters do not enter the first-order condition, shifting the level of $W$ by $\lambda V G$ but leaving the marginal tradeoff untouched. This is the paper's main qualitative message. It is exact for $\ell_0$ and holds up to an explicitly bounded shift for $\ell$ (Remark~\ref{rem:exactfoc}). Second, $\partial \dt^{\mathrm{opt}}/\partial c > 0$: more expensive consensus lengthens the optimum, the planner amortizing the per-block cost over longer intervals and accepting more LVR per block. Third, $\partial \dt^{\mathrm{opt}}/\partial \gamma$ is U-shaped, since at very low $\gamma$ the multiplier $F$ stays close to $1$ and block-shortening cannot suppress diffusion-LVR, while at very high $\gamma$ it is already crushed and shortening offers no further benefit; the interior minimum sits where shortening is an effective lever. Fourth, $\partial \dt^{\mathrm{opt}}/\partial V = 0$: since $c = r\,V/V_{\text{chain}}$ is proportional to $V$, the right-hand side of \eqref{eq:planner} equals $\gamma^2 V_{\text{chain}}/(8r)$ and $V$ cancels, so the optimum is a property of the chain and the fee tier, not of the pool. This one is exact: $V$ cancels from \eqref{eq:planner} identically, with no appeal to the separation.
\end{remark}

\begin{remark}[What the remainder does to the optimum]\label{rem:exactfoc}
The invariance in $(\lambda,m,\delta)$ is a property of $\ell_0$, whose jump term is an additive constant in $\dt$. In the exact model $\mathcal{E}$ carries both $\lambda$ and $\dt$, so the exact first-order condition is not jump-free and the invariance is not a theorem about $\ell$. Its failure is small enough to bound directly. Under \eqref{eq:mix}, any minimizer of $\ell+c/\dt$ lies in $\{\dt: W(\dt)+\mathcal{E}^-(\dt) \le \min_s [W(s)+\mathcal{E}^+(s)]\}$ with $\mathcal{E}^\pm$ the envelope of Proposition~\ref{prop:error}; at the calibration of Section~\ref{sec:numerics} this set is $[7.3,\,9.9]$~s around the separated optimum $\dt^{\mathrm{opt}}=8.4$~s. So the remainder can move the optimum by about a second --- which is also the range over which $W$ sits within $1$~bp/yr of its minimum anyway, so the shift is not resolvable against the flatness of the objective. The honest statement is therefore quantitative rather than absolute: jump parameters move $\dt^{\mathrm{opt}}$ by less than the flatness of the objective can resolve.
\end{remark}

\section{Numerical Illustration}\label{sec:numerics}

I illustrate the qualitative structure of Theorems~\ref{thm:fee} and \ref{thm:planner} for a hypothetical full-range LP on an ETH/USDT Uniswap pool. Model parameters are summarized in Table~\ref{tab:params}.

\begin{table}[htbp]
\centering
\caption{Baseline model parameters for the ETH/USDT full-range LP illustration}\label{tab:params}
\begin{tabular}{lc}
\toprule
Parameter & Value \\
\midrule
$\sigma$ (diffusion volatility) & $0.8156$ \\
$\lambda$ (jump intensity) & $283.3$/yr \\
$m$ (mean log jump) & $0$ \\
$\delta$ (std of log jump) & $0.0192$ \\
$\gamma$ (swap fee) & $0.0005$ \\
$V$ (pool TVL) & $1{,}000{,}000$ USD \\
$c$ (consensus cost) & $260 \times 10^{-5}$ USD/block \\
\bottomrule
\end{tabular}
\end{table}

The price parameters $(\sigma,\lambda,\delta)$ are estimated from Binance ETH/USDT $5$-minute closes over January 2020--June 2026, $2{,}373$ days. I take $\sigma$ from the continuous part of realised variance, estimated by bipower variation, and obtain $\lambda$ and $\delta$ from the jump-detection test of Lee and Mykland at the $1\%$ level, correcting the detected jump sizes for their diffusive component. The per-pool consensus cost $c$ follows from Ethereum's gross issuance ($\sim\!\$1.7$B/yr at $V_{\text{ETH}} \approx \$208$B), scaled by $25\%/30\% \approx 0.83$ to reflect that current staking ($\sim\!30\%$ of supply) exceeds Drake's security-optimal $1/4$~\cite{drake2025croissant,elowsson2024endgame}, then attributed pro-rata by $V/V_{\text{chain}} = 4.8 \times 10^{-6}$, yielding $c \approx \$260 \times 10^{-5}$/block. The remaining parameters are modeling choices: $m = 0$ (symmetric jumps --- the sample gives $\hat m = -0.0012$ with centred Lee--Mykland local-volatility window; the $99\%$ interval from the asymptotic standard error of the mean detected jump, $[-0.0026, +0.0003]$, contains zero, consistent with the choice), $\gamma = 5$~bps (Uniswap $0.05\%$ tier), $V = \$1$M (pool TVL). Sensitivities below cover $\sigma \in \{0.30, 0.816, 1.50\}$, $\lambda \in \{71, 283, 1133, 4533\}$/yr, $c \in \{2, 30, 90, 260, 500\} \times 10^{-5}$.

\paragraph{Closed-Form Magnitudes (Theorem~\ref{thm:fee}).}
Two reference numbers anchor this section: the diffusion ceiling $\sigma^2/8 = 832$~bp/yr (as $\dt \to \infty$) and the jump floor $\lambda G = 125$~bp/yr (as $\dt \to 0$). The fee has a negligible effect on the jump term ($\Psi(0.0260) = 0.959$, a $4.1\%$ trim) since $\gamma/\delta = 0.0260$ lies deep inside the support of $J$. On the diffusion term the fee is the dominant suppression mechanism at short $\dt$.

\paragraph{Chain-by-Chain LVR Rate.}
Evaluating \eqref{eq:main} at the baseline $(\sigma,\lambda,\delta,\gamma)$ over a range of $\dt$ from Ethereum L1 down to a sub-second app-chain target, with $\kappa(\dt) = \gamma/(\sigma\sqrt{\dt})$ and the diffusion share $(\sigma^2/8)F(\kappa)/\ell$ reported alongside, gives Table~\ref{tab:chains}. The $\sqrt{\dt}$ rate of Corollary~\ref{cor:floor} is visible: a $240\times$ cut from $12$~s to $50$~ms removes only $89\%$ of the diffusion term, leaving the rate $29\%$ above the floor. The two components are equal at $\dt_\times = 0.75$~s, and the diffusion term reaches a tenth of the floor only at $\dt = 5.5$~ms.

\begin{table}[htbp]
\centering
\caption{LVR rate $\ell$ at the baseline parameters across deployed block times, with the diffusion share $(\sigma^2/8)F(\kappa)/\ell$ alongside. Chain labels denote block times only; every row is the same baseline pool.}\label{tab:chains}
\begin{tabular}{lrrr}
\toprule
Chain & $\dt$ & $\ell$ (bp/yr) & diffusion share \\
\midrule
Ethereum L1 & $12$ s & $471$ & $73\%$ \\
Base / OP L2 & $2$ s & $312$ & $60\%$ \\
Solana & $400$ ms & $221$ & $43\%$ \\
Arbitrum & $250$ ms & $203$ & $38\%$ \\
App-chain & $50$ ms & $162$ & $23\%$ \\
\midrule
Jump floor ($\dt\to 0$) & --- & $125$ & $0\%$ \\
\bottomrule
\end{tabular}
\end{table}

\paragraph{Regime Sensitivity at Ethereum's $12$ s.}
Holding $\dt = 12$~s fixed, Table~\ref{tab:grid} reports $\ell$ across a $(\sigma,\lambda)$ grid, with the diffusion share of $\ell$ in brackets. The share is what determines whether shortening blocks reduces overall LVR: only the diffusion channel responds to $\dt$, so a high share means the pool has much to gain and a low one that the gain is small. The two readings move in opposite directions across the grid. The level of $\ell$ rises with both parameters, spanning $55$ to $3{,}597$~bp/yr, a factor of $65$; the diffusion share rises with $\sigma$ but falls with $\lambda$, from $1\%$ in the calm, jump-heavy corner to $98\%$ in the stressed, jump-light one. At the baseline $\lambda=283$/yr the share crosses $50\%$ at $\sigma = 0.55$, so the calibrated pool is firmly diffusion-dominated at $73\%$ and stress conditions reinforce this, while a calm regime puts the slot where shorter blocks help little. Any policy claim about $12$-second blocks is therefore conditional on the volatility-and-jump-intensity regime, not on a single point estimate.

\begin{table}[htbp]
\centering
\caption{LVR rate $\ell$ in bp/yr at $\dt = 12$~s across the $(\sigma,\lambda)$ grid, with the diffusion share of $\ell$ in brackets}\label{tab:grid}
\begin{tabular}{lrrrr}
\toprule
$\sigma \,\backslash\, \lambda$ & $71$/yr & $283$/yr & $1{,}133$/yr & $4{,}533$/yr \\
\midrule
$0.30$ (calm)      & $55$ $(43\%)$      & $148$ $(16\%)$     & $524$ $(4\%)$      & $2{,}027$ $(1\%)$ \\
$0.816$ (baseline) & $377$ $(92\%)$     & $471$ $(73\%)$     & $847$ $(41\%)$     & $2{,}349$ $(15\%)$ \\
$1.50$ (stress)    & $1{,}625$ $(98\%)$ & $1{,}719$ $(93\%)$ & $2{,}095$ $(76\%)$ & $3{,}597$ $(44\%)$ \\
\bottomrule
\end{tabular}
\end{table}

\paragraph{Planner's Optimum.}
With $c$ calibrated to the issuance-derived baseline $260 \times 10^{-5}$ USD/block per \$1M TVL, \eqref{eq:planner} gives $A = V\gamma^2/(8c) = 12.02$ and $\eta^{\mathrm{opt}} = 1.677$. Across the sensitivity range, $c = (500, 260, 90, 30, 2)\times10^{-5}$ USD/block yields $\dt^{\mathrm{opt}} = (15.4,\, 8.4,\, 3.4,\, 1.4,\, 0.20)$~s respectively --- nearly two orders of magnitude, driven by the cost calibration. Volatility moves it comparably: $\sigma \in \{0.30, 0.815, 1.50\}$ gives $\dt^{\mathrm{opt}} = (62.3,\, 8.4,\, 2.5)$~s, since $\dt^{\mathrm{opt}} \propto \gamma^2/\sigma^2$ at fixed $\eta^{\mathrm{opt}}$. At the baseline $c = 260 \times 10^{-5}$, $\dt^{\mathrm{opt}} \approx 8.4$~s or roughly two-thirds of Ethereum's block-time and well above sub-second L2/app-chain timing. Under an LP-LVR-only objective Ethereum is too slow, though the gain is modest: $W$ falls from $539$ to $533$~bp/yr, a saving of $6$~bp/yr, under $2\%$, and stays within $1$~bp/yr of its minimum over $\dt \in [7.4, 9.7]$~s. What moves the optimum is $c$, $\sigma$ and $\gamma$; the jump parameters do not move it at all. The conclusion is also one-sided: a full social-welfare treatment with MEV, finality, and oracle freshness could push $\dt^{\mathrm{opt}}$ back up.

\section{Discussion and Conclusion}\label{sec:disc}

I extend the LVR framework to an empirical feature of crypto reference prices that the block-time literature has so far left out: price discontinuities. Jumps turn out to change what block time can achieve, not what it should be. For the CPMM under jump-diffusion prices the LVR rate splits into a diffusion channel that the block schedule governs through the fee multiplier $F(\kappa)$ and a $\dt$-invariant jump channel that it cannot touch. The planner's optimum solves an explicit cubic and is invariant in pool size exactly and in the jump parameters $(\lambda,m,\delta)$ up to a shift the remainder bound confines to about a second (Remark~\ref{rem:exactfoc}), so jumps shift the level of LVR but not, to any resolution the objective can support, the marginal tradeoff that sets $\dt^{\mathrm{opt}}$.

\paragraph{Policy.}
Shortening blocks can at most close the gap between the current rate and the jump floor $\lambda VG$ of Theorem~\ref{thm:floor}: the floor bounds the rate from below, so no block-time reduction saves more than $\ell(\dt) - \lambda VG$, and because $F$ decays algebraically rather than with a Gaussian tail, that gap closes only as $\sqrt{\dt}$ (Corollary~\ref{cor:floor}). At the calibration of Section~\ref{sec:numerics} the diffusion channel does not become negligible until the low-millisecond range, so at every deployed block time in Table~\ref{tab:chains} what limits block-time policy is not the jump floor but the cost of consensus. Once the floor does bind, block-time refinement no longer helps, and Remark~\ref{rem:feerecovery} identifies the lever that remains: since fees recover only the fraction $1-\Psi(\gamma/\delta)$ of jump-LVR, the binding ratio is $\gamma/\delta$. The LP-side levers are therefore fee-tier design, batch auctions that eliminate the per-block arbitrage window at source, and LVR-rebate designs such as auction-managed AMMs that route pool-management revenue back to LPs.

Calibrating the consensus cost from Ethereum's gross issuance places the LP-side planner's optimum below the current $12$-second block-time but well above the sub-second regime. Two caveats temper the reading: the gain over the status quo is under $2\%$ of the objective, and $\dt^{\mathrm{opt}}$ is sensitive to the cost calibration---invariant in pool size and effectively in every jump parameter (Remarks~\ref{rem:cs}--\ref{rem:exactfoc}), it moves by nearly two orders of magnitude across the plausible range of $c$. The qualitative takeaway is that Ethereum's block-time is somewhat too long under an LP-LVR-only objective, while the aggressive sub-second targets pursued by some chains overshoot what LVR considerations alone would justify.

\paragraph{Limitations.}
The welfare objective is LP-LVR-only, so the optimum from $W(\dt) = \ell_0(\dt) + c/\dt$ is an LP-side input to a larger welfare problem, not a prescription: a full treatment must also weigh MEV redistribution, finality, censorship resistance, and oracle freshness. The consensus cost attributes issuance pro-rata by value secured, abstracting from the dilution a native-token-holding LP bears directly, and equates observed issuance with validator operating cost --- which breaks where issuance is a growth subsidy or validators are funded by MEV and fees. The pool is full-range, whereas most LVR is paid on concentrated liquidity, where the two effects on $G$ run in opposite directions: concentration raises the pool curvature per unit of TVL and so inflates $G$ for jumps contained in the active range, while a jump large enough to cross the range leaves the position single-asset and truncates the loss. Which dominates depends on $\delta$ relative to the range width, so the full-range floor is not simply a lower bound for concentrated positions. The arbitrageur has unlimited capital, zero latency and continuous sight of $S_t$, an idealization weakest in the sub-second regime the floor argument concerns. Blocks finally arrive as a Poisson process rather than on a deterministic proof-of-stake schedule, a tractability assumption inherited from~\cite{milionis2023automated} that is conservative, since the deterministic schedule minimizes arbitrage profit~\cite{nezlobin2025lvr}.

The practical message is that the jump floor is genuine --- for symmetric jump laws Theorem~\ref{thm:floor} establishes it as an exact lower bound, not an artifact of the separation --- but distant: at deployed block times the binding consideration is what a block costs, and the lever against the floor itself is the ratio of fee to jump size, not the block schedule.

\bibliographystyle{splncs04}
\bibliography{references}

\ifextended
\appendix
\section{Omitted Proofs}\label{app:proofs}

This appendix collects the proofs deferred from the main text. Statements are
those of the corresponding numbered results above.

\subsection{Proof of Theorem~\ref{thm:decomp}}\label{app:decomp}

\begin{proof}
The plan: apply the It\^o--L\'evy formula to $V(S_t)$, identify \eqref{eq:lvrdef} as the sum of a quadratic-variation residual and a jump-concavity residual, both nonnegative, then evaluate each for the CPMM.

Write $V_t := V(S_t)$. The It\^o--L\'evy formula for jump-diffusion processes \cite[Prop.~8.14]{cont2003financial} applied to $V$, with $S$ satisfying \eqref{eq:merton}, gives
\begin{align*}
V(S_t) - V(S_0) =\;& \int_0^t V'(S_{s-})\,dS_s \\
&+ \tfrac{1}{2}\int_0^t V''(S_{s-})\,\sigma^2 S_{s-}^2 \,ds \\
&+ \sum_{0 < s \leq t} \!\Big[V(S_s) - V(S_{s-}) - V'(S_{s-})(S_s - S_{s-})\Big].
\end{align*}
Substituting into \eqref{eq:lvrdef},
\begin{equation}\label{eq:lvr-itl}
\mathrm{LVR}(t) \;=\; -\tfrac{1}{2}\int_0^t V''(S_{s-})\,\sigma^2 S_{s-}^2 \, ds \;-\; \sum_{0 < s \leq t}\!\Big[V(S_s) - V(S_{s-}) - V'(S_{s-})\,\Delta S_s\Big].
\end{equation}
Both terms are nonnegative because $V$ is concave: $V''<0$ and $V(S_s) \leq V(S_{s-}) + V'(S_{s-})\Delta S_s$.

\emph{Diffusion term.} With $V''(P) = -x(P)/(2P)$ and $V(P) = 2\sqrt{LP}$, $-\tfrac{1}{2} V''(P) \sigma^2 P^2 = \tfrac{\sigma^2}{4} P\,x(P) = \tfrac{\sigma^2}{4} \sqrt{L P} = \tfrac{\sigma^2}{8} V(P)$.

\emph{Jump term.} At a jump time $\tau$ with $S_\tau = S_{\tau-} e^{J}$, writing $r := e^J$,
\begin{align*}
V(S_\tau) - V(S_{\tau-}) - V'(S_{\tau-})\,\Delta S_\tau
&= 2\sqrt{L S_{\tau-}}\Big[\sqrt{r} - 1 - \tfrac{r-1}{2}\Big] \\
&= -\sqrt{L S_{\tau-}}\,(\sqrt{r}-1)^2 \;=\; -\tfrac{1}{2}\,V(S_{\tau-})\,(e^{J/2} - 1)^2.
\end{align*}
The jump contribution to \eqref{eq:lvr-itl} is therefore $\sum_{\tau} \tfrac{1}{2}\,V(S_{\tau-})\,(e^{J_\tau/2} - 1)^2$. By the marked-Poisson compensator formula \cite[Prop.~8.8]{cont2003financial}, and using the independence of jump marks $J$ from the predictable filtration $\mathcal{F}_{\tau-}$,
\[
\E\!\Big[\sum_{\tau \leq t} \tfrac{1}{2}V(S_{\tau-})(e^{J_\tau/2}-1)^2\Big] \;=\; \tfrac{\lambda}{2}\,\E[(e^{J/2}-1)^2]\cdot\E\!\Big[\int_0^t V(S_{s-})\,ds\Big],
\]
so the expected rate at state $V_t$ is $\tfrac{\lambda}{2} V_t \E[(e^{J/2}-1)^2]$.

\emph{Moments.} Expanding $(e^{J/2}-1)^2 = e^{J} - 2 e^{J/2} + 1$ and taking expectations under $J \sim \mathcal{N}(m,\delta^2)$ gives the bracket in Theorem~\ref{thm:decomp}, nonnegative as the expectation of a square and vanishing only in the degenerate case $J\equiv 0$.
\end{proof}

\subsection{Proof of Lemma~\ref{lem:F}}\label{app:F}

\begin{proof}
Only the tails of \eqref{eq:pi0} contribute, since $(|z|-\gamma)_+ = 0$ on the band. Substituting the tail density $c\,e^{-\theta(|z|-\gamma)}$ into \eqref{eq:rate-stationary} with the quadratic loss, and writing $u := |z|-\gamma$ for the exceedance,
\[
\ell_{\mathrm{diff}} \;=\; \frac{V}{8\dt}\,\E_{\pi}\big[(|z|-\gamma)_+^2\big]
\;=\; \frac{V}{8\dt}\cdot 2\int_0^\infty u^2\,c\,e^{-\theta u}\,du
\;=\; \frac{V}{8\dt}\cdot\frac{4c}{\theta^3}
\;=\; \frac{V}{8\dt}\cdot\frac{2}{\theta^2(1+\eta)},
\]
the last step using $c=\theta/(2(1+\eta))$. Since $\theta^2 = 2/(\sigma^2\dt)$, the prefactor collapses, $2/(\theta^2\dt) = \sigma^2$, leaving
$\ell_{\mathrm{diff}} = \tfrac{\sigma^2V}{8}(1+\eta)^{-1}$ with $\eta=\sqrt2 \kappa$. The multiplier $(1+\eta)^{-1}$ is precisely the tail mass of \eqref{eq:pi0}, which is the identification $F = \Prob_\pi(|z|>\gamma)$; the stated properties of $F$ are elementary.
\end{proof}

\subsection{Proof of Lemma~\ref{lem:G}}\label{app:G}

\begin{proof}
The plan: obtain the jump count over an inter-block interval, compute the loss a single jump inflicts once the fee absorbs its first $\gamma$, and then name the two ways the separated charge $\lambda VG$ departs from the exact one-jump contribution --- both deferred to Proposition~\ref{prop:error}.

\emph{Jump count.} Jumps and blocks are independent Poisson processes, so the jump count $N$ over an inter-block interval is geometric: with $\rho:=\lambda\dt/(1+\lambda\dt)$, exactly one jump has probability $(1-\rho)\rho$ and two or more probability $\rho^2$, both exactly. 

\emph{Per-jump loss.} Conditional on one jump of mark $J$ with $|J|>\gamma$, the arbitrageur executes at the next block, leaving the pool's marginal log-price at distance $\gamma$ from the external price and capturing the residual concavity gap; for the quadratic surrogate of \eqref{eq:rate-stationary} (with $\xi := |J|-\gamma$) this loss is $\tfrac{V}{8}(|J|-\gamma)^2$. Averaging over the mark gives $\tfrac{V}{8}\E[g(J)] = V\,G(\gamma;m,\delta^2)$ per one-jump interval.

\emph{The two approximations.} The separated charge $\lambda V G$ differs from the exact one-jump contribution in two ways, both bounded in Proposition~\ref{prop:error}. First, the pre-block mispricing is in fact $z = M + X + J$, with $M\in[-\gamma,\gamma]$ the state carried in and $X$ the interval's diffusive increment; write $Y := M+X$ for the pair. Dropping $Y$ is the interaction deficit. Second, one-jump intervals occur at the thinned block rate $(1-\rho)\rho/\dt = \lambda(1+\lambda\dt)^{-2}$, so the exact one-jump contribution is $\lambda(1+\lambda\dt)^{-2}\,V\,G(\gamma;m,\delta^2)$ before the interaction correction; charging instead at the full jump rate $\lambda$ --- which is what keeps the term free of $\dt$ --- is the thinning deficit.

Strict monotonicity in $\gamma$ follows since $(|j|-\gamma)^2 \1\{|j|>\gamma\}$ is strictly decreasing in $\gamma$ pointwise on the support.
\end{proof}

\subsection{Proof of Proposition~\ref{prop:error}}\label{app:error}

\begin{proof}
The plan: condition on the jump count $N$ over an inter-block interval and bound each term of the display in turn --- the multi-jump term ($N\ge2$), the interaction and the thinning (the two deficits on $N=1$, identified in the proof of Lemma~\ref{lem:G}), the block clock and the stationary law (the two deficits on $N=0$), and the curvature residual. Throughout, write $u:=\lambda\dt$ and $\rho:=u/(1+u)$, so that $N$ is geometric, $\Prob(N{=}k)=\rho^k(1-\rho)$.

\emph{Multi-jump term ($N\ge2$).} On an interval with $k\ge2$ jumps, $g(z)\le z^2$ gives a loss at most $\tfrac{V}{8}\big(\E[M^2]+\E[X^2\mid N{=}k]+k\delta^2\big)$, with $\E[M^2]\le\gamma^2$ pathwise and, since $T\mid N{=}k\sim\mathrm{Gamma}(k{+}1,\dt^{-1}{+}\lambda)$, $\E[X^2\mid N{=}k]=\sigma^2(k{+}1)\dt/(1+u)$. The geometric sums over $k\ge2$ are
\[
\textstyle\sum\rho^k(1-\rho)=\rho^2, \qquad
\textstyle\sum k\rho^k(1-\rho)=\frac{\rho^2(2-\rho)}{1-\rho}, \qquad
\textstyle\sum (k{+}1)\rho^k(1-\rho)=\frac{\rho^2(3-2\rho)}{1-\rho}.
\]
Dividing by $\dt$ and relaxing via $\rho^2/\dt\le\lambda^2\dt$, $\rho^2(2-\rho)/\big((1-\rho)\dt\big)\le2\lambda^2\dt$ and $\rho^2(3-2\rho)/\big((1-\rho)(1+u)\dt\big)\le3\lambda^2\dt$ bounds the contribution by $\tfrac{V\lambda^2\dt}{8}(\gamma^2+3\sigma^2\dt+2\delta^2)$, the last term of the display.

\emph{The single-jump split ($N=1$).} Conditional on exactly one jump, write the pre-block mispricing as $Y+J$, with $Y := M+X$ independent of the mark $J$. Convexity of $g$ with $2$-Lipschitz gradient gives $0\le g(y+j)-g(j)-g'(j)y\le y^2$ pointwise, so taking expectations,
\[
0 \;\le\; \E[g(Y+J)]-\E[g(J)] - \E[g'(J)]\,\E[Y] \;\le\; \E[Y^2\mid N{=}1].
\]
Symmetry gives $\E[Y]=0$, so the cross term drops. The single-jump contribution to the rate is $\tfrac{V}{8}\,\lambda(1+u)^{-2}\,\E[g(Y+J)]$, since $\Prob(N{=}1)/\dt=\lambda(1+u)^{-2}$, against the separated charge $\tfrac{V}{8}\,\lambda\,\E[g(J)]$; their difference splits as
\[
\tfrac{V}{8}\,\lambda(1+u)^{-2}\big(\E[g(Y{+}J)]-\E[g(J)]\big) \;-\; \tfrac{V}{8}\,\lambda\big[1-(1+u)^{-2}\big]\,\E[g(J)].
\]
The first piece is the interaction, the second the thinning deficit; each is bounded in turn.

\emph{Interaction.} It is \emph{nonnegative} by the first display, and since $(1+u)^{-2}\le1$ it is at most $\tfrac{\lambda V}{8}\E[Y^2\mid N{=}1]$. For its moments, conditioning on a jump size biases the interval, $T\mid N{=}1\sim\mathrm{Gamma}(2,\dt^{-1}+\lambda)$, whence $\E[X^2\mid N{=}1]\le2\sigma^2\dt$, while $\E[M^2]\le\gamma^2$ because the post-block state lies in the band by construction --- no property of its stationary law is used.

\emph{Thinning.} It is nonpositive, and $1-(1+u)^{-2}\le2u$ bounds it below by $-2\lambda^2\dt\,VG_\nu(\gamma)$, the third term on the lower side. The jumps that thinning removes from single-jump intervals land in multi-jump intervals, whose losses are carried, nonnegative, on the upper side.

\emph{Block clock ($N=0$).} Conditioning on a jump-free interval gives weight $(1+u)^{-1}$ and makes the inter-block increment Laplace with rate $\theta\sqrt{1+u}$, so by the computation of Lemma~\ref{lem:F} that contribution is exactly $\ell_{\mathrm{diff}}\,R$ with
\[
R \;=\; \frac{1+\eta}{(1+u)^2\big(1+\eta\sqrt{1+u}\big)}.
\]
Here $R\le1$, since $(1+u)^2\ge1$ and $1+\eta\sqrt{1+u}\ge1+\eta$. For the other side, $(1+u)^{-2}\ge1-2u$ and $\sqrt{1+u}\le1+u/2$ give $(1+\eta)/(1+\eta\sqrt{1+u})\ge1-\tfrac{\eta u}{2(1+\eta)}\ge1-u/2$, and $(1-2u)(1-u/2)=1-\tfrac52u+u^2\ge1-\tfrac52u$. The block-clock deficit therefore lies between $0$ and $\tfrac52\lambda\dt\,\ell_{\mathrm{diff}}$.

\emph{Stationary law.} The $N=0$ computation above still evaluates the carried state against the $\lambda=0$ law; jump feedback shifts that law, and this term bounds the shift. The plan, in four steps: (a) bound the one-step kernel distance by coupling; (b) convert it into a stationary-law distance via a mixing time for $K_0$; (c) bound that mixing time by condition \eqref{eq:mix}; (d) convert the law shift into a rate shift through the oscillation of the per-block loss.

\emph{(a) Coupling.} Couple the two post-block chains on the same block times and Brownian path: they agree unless a jump lands, and $\Prob(N\ge1)=1-\E[e^{-\lambda T}]=\lambda\dt/(1+\lambda\dt)$ for $T\sim\mathrm{Exp}(1/\dt)$, so $\sup_m\|K(m,\cdot)-K_0(m,\cdot)\|_{\mathrm{TV}}\le\lambda\dt$ and, telescoping, $\|K^n-K_0^n\|\le n\lambda\dt$.

\emph{(b) Mixing.} If $n_*$ is a mixing time for $K_0$, in the sense $\sup_{m,m'}\|K_0^{n_*}(m,\cdot)-K_0^{n_*}(m',\cdot)\|\le\tfrac12$, then combining the two bounds at $n=n_*$ yields $\|\pi-\pi_0\|\le2n_*\lambda\dt$. Everything therefore turns on $n_*$.

\emph{(c) Condition \eqref{eq:mix}.} The condition is the form the band's relaxation takes: an interval of width $2\gamma$ carrying a walk of step variance $2/\theta^2$ has diffusive relaxation time $4\eta^2/\pi^2$, and the additive $1$ keeps it meaningful for the integer $n_*$ at small $\eta$. Establishing \eqref{eq:mix} for the clipped Laplace kernel itself --- sticky rather than reflecting boundaries, discrete rather than diffusive at finite $\eta$ --- is not attempted here. Since $K_0$ depends on $(\gamma,\sigma,\dt)$ only through $\eta$, and not at all on $\lambda$, \eqref{eq:mix} is a condition on $\eta$ alone, and it is checked numerically over the whole empirically relevant range: every block time from $50$~ms to $12$~s, every volatility in $[0.30,1.50]$, every jump intensity, and every Uniswap fee tier from $1$ to $100$~bps --- jointly, $\eta$ from $0.15$ to $1.2\times10^{3}$. Over that range the ratio $n_*/(1+\eta^2)$ peaks at $1.27$ near $\eta=0.8$, against the $2$ that \eqref{eq:mix} asks for, and settles near $0.38$ in the fast-block, wide-band corner, so the margin is smallest where the band is narrow relative to a block's diffusion and only widens from there. Substituting $\eta^2=2\gamma^2/(\sigma^2\dt)$ into \eqref{eq:mix} gives $\|\pi-\pi_0\|\le 2C\lambda\dt+4C\lambda\gamma^2/\sigma^2$; note the $\dt$ cancels in the second piece, which is why this term does not vanish with the block time.

\emph{(d) From law shift to rate shift.} The per-block diffusion loss as a function of the carried state, $\phi(m)=\tfrac{V}{8}\E[(|m+X|-\gamma)_+^2] = \tfrac{V}{8\theta^2}\,2e^{-\eta}\cosh(\theta m)$, has $\mathrm{osc}(\phi)=\tfrac{V}{8\theta^2}(1-e^{-\eta})^2\le V\sigma^2\dt/16$, and $\mathrm{osc}(\phi)\,\|\pi-\pi_0\|/\dt \le \tfrac{\lambda V}{4}(2\gamma^2+\sigma^2\dt)$ at $C=2$.

\emph{Curvature.} Finally $\mathcal{E}_{\mathrm{curv}} = \tfrac{V}{2}\E_\pi[\rho(\xi)]$ is nonnegative by \eqref{eq:rhosign}, $\pi$ being symmetric.
\end{proof}

\subsection{Proof of Theorem~\ref{thm:floor}}\label{app:floor}

\begin{proof}
The plan: establish the chain
\[
\E_\pi[h(\xi)] \;\ge\; \tfrac{V}{8}\,\E_\pi[g(z)] \;\ge\; \tfrac{V}{8}\,\E\big[g\big(\textstyle\sum_1^N J\big)\big] \;\ge\; \tfrac{V}{8}\,\E[N]\,\E[g(J)],
\]
three inequalities in the same direction, established as (a)--(c) below; dividing by $\dt$ and using $\E[N]=\lambda\dt$ then finishes. Throughout, $N$ is the jump count over an inter-block interval, $J_1,\dots,J_N$ are the marks, and the pre-block mispricing is $z = Y + \sum_{k\le N}J_k$.

\emph{(a) Concavity residual.} At $m=0$ the law of $z$ is symmetric, so $\E_\pi[h(\xi)] \ge \tfrac{V}{8}\E_\pi[\xi^2] = \tfrac{V}{8}\E_\pi[g(z)]$ by \eqref{eq:rhosign}. It therefore suffices to bound the convex surrogate from below.

\emph{(b) Interaction.} All the next step needs is $\E[Y\mid N{=}k,\,\sum J = s]=0$. $M$ is fixed by the previous block and so is independent of this interval's jumps, with $\E[M]=0$ by symmetry; the Brownian path is independent of the Poisson process and its marks, and $X\mid T\sim\mathcal N(0,\sigma^2T)$ is mean-zero conditionally on $T$, a property preserved by averaging over any law of $T$, in particular the size-biased one induced by conditioning on $N$. Conditional Jensen applied to the convex $g$ then gives $\E[g(Y+s)\mid N{=}k,\,\sum J = s] \ge g(s)$, hence $\E[g(z)\mid N{=}k] \ge \E[g(\sum_{k}J)]$. Dropping the interaction can only lower the loss.

\emph{(c) Aggregation.} This is the hypothesis, and Merton satisfies it: $\sum_{1}^{k}J \sim \mathcal N(0,k\delta^2)$, so by \eqref{eq:Gclosed} and the monotonicity of $\Psi$,
\[
\E\big[g(\textstyle\sum_1^k J)\big] \;=\; k\delta^2\,\Psi\!\big(\tfrac{\gamma}{\delta\sqrt k}\big) \;\ge\; k\delta^2\,\Psi\!\big(\tfrac{\gamma}{\delta}\big) \;=\; k\,\E[g(J)],
\]
since $\gamma/(\delta\sqrt k)\le\gamma/\delta$. Several jumps in one interval are cleared by a single trade, and that trade costs at least as much as clearing them one at a time.

Combining, and using $\E[N]=\lambda\dt$ exactly for a Poisson jump process observed over an independent $\mathrm{Exp}(1/\dt)$ interval,
\[
\ell(\dt) \;=\; \tfrac{1}{\dt}\E_\pi[h(\xi)] \;\ge\; \tfrac{V}{8\dt}\sum_{k\ge0}\Prob(N{=}k)\,k\,\E[g(J)] \;=\; \tfrac{V}{8}\,\tfrac{\E[N]}{\dt}\,\E[g(J)] \;=\; \lambda V G .
\]
Both (a) and (b) use symmetry, and neither survives $m\neq0$ unchanged: $\rho$ is negative for downward moves, so the surrogate step is signed only against a symmetric law, and (b) picks up $\E[g'(\sum J)]\E[Y]$. Proposition~\ref{prop:error} carries both deficits for the rate itself.
\end{proof}

\subsection{Proof of Corollary~\ref{cor:floor}}\label{app:corfloor}

\begin{proof}
$F$ is strictly decreasing by Lemma~\ref{lem:F} and $\kappa = \gamma/(\sigma\sqrt{\dt})$ is strictly decreasing in $\dt$, so $F(\kappa(\dt))$ is strictly increasing in $\dt$; the limits follow from $F(\infty)=0$ and $F(0)=1$, and the rate from $F(\kappa)\sim\sigma\sqrt{\dt}/(\sqrt2\gamma)$. The jump term is independent of $\dt$. The transfer to $\ell$ follows from $\ell(\dt_2)-\ell(\dt_1) \ge \ell_0(\dt_2)-\ell_0(\dt_1) - w$, where $w = 0.849+0.653 = 1.50$~bp/yr is the width of the remainder envelope of Proposition~\ref{prop:error} over $\dt\le12$~s.
\end{proof}

\subsection{Proof of Theorem~\ref{thm:planner}}\label{app:planner}

\begin{proof}
The plan: reduce the first-order condition to a function of $\dt$ that is strictly monotone, which gives existence and uniqueness of the minimum in one stroke, then solve the resulting cubic in $\eta$ by Cardano.

$\ell_0$ is differentiable and strictly increasing; $c/\dt$ is strictly decreasing. Computing $W'(\dt) = \ell_0'(\dt) - c/\dt^2$ and equating to zero gives the first-order condition $\ell_0'(\dt^{\mathrm{opt}}) = c/(\dt^{\mathrm{opt}})^2$. Substituting \eqref{eq:focexplicit} and $\sigma^2\dt = \gamma^2/\kappa^2$,
\[
\dt^2\,W'(\dt) \;=\; q(\dt) - c, \qquad q(\dt) := \dt^2\ell_0'(\dt) = \frac{V\gamma^2}{8\,\eta\,(1+\eta)^2},\quad \eta = \sqrt2\,\kappa .
\]
Now $q$ is strictly decreasing in $\eta$ and hence strictly increasing in $\dt$, with $q\to0$ as $\dt\to0$ ($\eta\to\infty$) and $q\to\infty$ as $\dt\to\infty$ ($\eta\to0$). So $W'$ changes sign exactly once, from negative to positive, and $W$ has a unique interior global minimum, at which $q = c$; that is \eqref{eq:planner}. For the closed form of $\eta^{\mathrm{opt}}$, substituting $\eta = w-2/3$ into $\eta^3+2\eta^2+\eta-A=0$ gives the depressed cubic $w^3-w/3-(2/27+A)=0$, whose polynomial discriminant is negative for every $A>0$, leaving one real root $w = u^{1/3}+u'^{1/3}$ with $uu'=1/729$; evaluating the smaller cube root as $\tfrac19u^{-1/3}$ avoids cancellation.
\end{proof}

\section{Illustrations}\label{app:figs}

Both figures are drawn at stylised parameters chosen for legibility; the
calibrated separation of Section~\ref{sec:numerics} is more extreme
($\delta/\gamma \approx 38$, $\lambda\dt \approx 10^{-4}$ at $12$~s).

\begin{figure}[htbp]
\centering
\includegraphics[width=\textwidth]{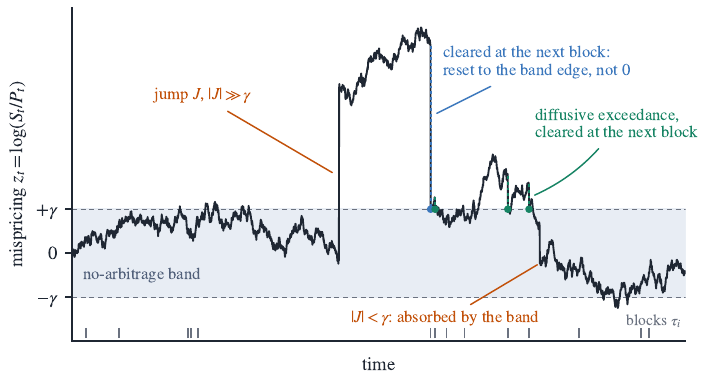}
\caption{A sample path of the mispricing $z_t=\log(S_t/P_t)$ under
\eqref{eq:merton}. Between Poisson block arrivals (ticks) $z$ diffuses freely;
at a block the arbitrageur trades iff $|z|>\gamma$, resetting $z$ to the band
edge $\pm\gamma$, not to $0$. Diffusive exceedances of the band (teal) are on
the scale $\sigma\sqrt{\dt}$ and are cleared at the next block. A jump with
$|J|\gg\gamma$ (orange) opens a $\dt$-independent gap, carried until the next
block and cleared in full (blue), while a jump with $|J|<\gamma$ is absorbed by
the band without a trade.}
\label{fig:zpath}
\end{figure}

\begin{figure}[htbp]
\centering
\includegraphics[width=\textwidth]{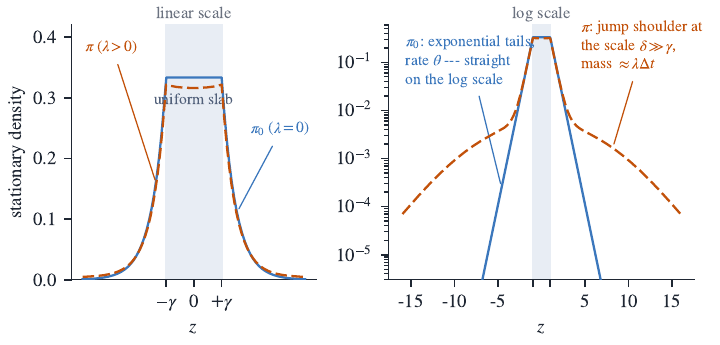}
\caption{The stationary law of the pre-block mispricing on a linear (left) and
logarithmic (right) scale, at $\eta=2$, $\delta=5\gamma$, $\lambda\dt=0.08$.
The law $\pi_0$ ($\lambda=0$, the closed form \eqref{eq:pi0}) is uniform on the
band with exponential tails of rate $\theta$. With Merton jumps, $\pi$
(quasi-exact: jump feedback on the carried state is ignored, an $O(\lambda\dt)$
effect) adds a far field at the jump scale $\delta\gg\gamma$ --- the shoulder
leaving the exponential tail on the log scale, mass that no block schedule can
pull in. The slab of $\pi$ is also slightly depressed towards its centre:
exact uniformity on the band is a $\lambda=0$ property, and jumps break it at
$O(\lambda\dt)$, exaggerated here by the stylised $\lambda\dt$.}
\label{fig:statlaw}
\end{figure}
\fi

\end{document}